\documentclass{ifacconf}
\usepackage{tikz}
\usepackage{graphicx,color} 
\usepackage{natbib} 
\usepackage{amsmath}
\usepackage{amssymb}
\usepackage{mathtools}

\newcommand{\CC}{\mathbb{C}}

\newcommand{\RR}{\mathbb{R}}

\def\A{\mathcal{A}}  
\def\B{\mathcal{B}}  
\def\C{\mathcal{C}} 
\def\D{\mathcal{D}} 
 
\def\G{\mathcal{G}}  
\def\L{\mathcal{L}} 
\def\K{\mathcal{K}}

\def\S{\mathcal{S}}

\def\R{\mathcal{R}}

\DeclareMathOperator{\rank}{rank}

\newtheorem{definition}{Definition}
\newtheorem{example}{Example}
\newtheorem{lemma}{Lemma}
\newtheorem{proposition}{Proposition}

\usepackage{float}

\usepackage{enumitem}
\setlist[itemize]{leftmargin=*}
\setlist[enumerate]{leftmargin=*}
\usepackage{comment}
\usepackage{array}

\begin{document}
\begin{frontmatter}

\title{Observer-Based Stabilization for\\Linear Multi-Agent Dynamical Systems\\
Using Generalized Frequency Variables}

\author[First]{G. Q. Bao Tran}
\author[Second]{Yutaka Hori}
\author[Third]{Shinji Hara}

\address[First]{Coordinated Science Laboratory, University of Illinois Urbana-Champaign, Urbana, IL 61801, USA\\(e-mail: baotran@illinois.edu).}
\address[Second]{Applied Physics and Physico-Informatics, Keio University, 3-14-1 Hiyoshi, Kohoku-ku, Yokohama, Kanagawa 223-8522, Japan\\(e-mail: yhori@appi.keio.ac.jp).}
\address[Third]{Supercomputing Research Center, Institute of Integrated Research, Institute of Science Tokyo, 2-12-1 Ookayama, Meguro-ku, Tokyo, Japan (e-mail: shinji\_hara@ipc.i.u-tokyo.ac.jp).}

\begin{abstract} 
We address the conditions and design of controllers and observers for homogeneous networks of linear MIMO agents. We develop networked controllers and observers that ensure the stability of both the system state and the estimation error, leveraging the concept of generalized frequency variables. A separation principle for networks is then established, showing that the observer and controller can be designed independently and combined to achieve a stable output feedback. Our results are illustrated via a highly unstable, oscillatory network of locally actuated pendulums on carts. Finally, necessary conditions for controllability and observability—derived from agent properties and network structure—are established and discussed.
\end{abstract}

\end{frontmatter}

\section{Introduction}
Multi-agent systems---those viewed as an interconnection of (a large number of) sub-systems or agents---are of particular research interest thanks to their numerous applications, e.g., in drones~\citep{souli}. To cope with the challenges arising in such large-scale networks, many studies have been conducted, e.g., on decentralized control~\citep{wangDecentralized,bakule}. One approach that can provide a unifying theoretical paradigm for a general class of multi-agent systems—where agents operate autonomously while exchanging information with one another—is the notion of \emph{generalized frequency variables} introduced in~\citep{haraCDC07}. In this viewpoint, a network of identical linear single-input--single-output (SISO) agents with transfer function $h(s)$ can be described by the transfer function $\G(s) := G(\phi(s))$, where $G(s)$ is a proper rational function representing only the interactions between agents, or the network structure, and $\phi(s) := 1/h(s)$. Within this framework, $\G(s)$ is obtained by simply replacing the traditional frequency $s$ in $G(s)$ with $\phi(s)$, motivating us to call $\phi(s)$ a generalized frequency variable. This paradigm, extended to the case of multi-input--multi-output (MIMO) agents in~\citep{haraRobust}, enables a separation of agent and network properties in the analysis and reduces the computational burden.

In this work, we study conditions for, and the design of, observers and controllers for homogeneous networks of linear time-invariant (LTI) MIMO agents using the language of generalized frequency variables. The proposed results are illustrated through a pendulum network example. Controllability and observability conditions for networks studied in~\citep{haraSICE,trumpf} are revisited for improved checkability.

\emph{Notations:} Denote by $\RR$ (resp., $\CC$) the set of real (resp., complex) numbers and by $\RR^{m \times n}$ (resp., $\CC^{m \times n}$) the set of real-valued (resp., complex-valued) $(m\times n)$ matrices. Let $\R_p$ be the set of real rational functions. Let $I_n \in \RR^{n \times n}$ be the identity matrix. Let $A \otimes B$ be the Kronecker product of matrices $A$ and $B$. Denote $\sigma(A)$ as the set of eigenvalues of the square matrix $A$. A polynomial is Hurwitz if all its roots are in the open left half complex plane.

\section{System Representation}
Consider a homogeneous network of $N$ identical $n$-dimensional LTI MIMO agents of the form
\begin{equation}\label{eq:sysi}
\dot{x}_i = A_h x_i + B_h u_i, \qquad y_i = C_h x_i,
\end{equation}
i.e., described by the matrices $(A_h,B_h,C_h) \in \RR^{n \times n} \times \RR^{n \times m} \times \RR^{m \times n}$, with interconnection structure characterized by the matrices $(A,B,C) \in \RR^{N \times N} \times \RR^{N \times M} \times \RR^{M \times N}$. From~\eqref{eq:sysi}, the matrix transfer function of each agent is
\begin{equation}\label{eq:Hs}
H(s) = C_h (s I_n - A_h)^{-1} B_h \in \R_p^{m\times m},
\end{equation}
where $s \in \CC$ denotes the frequency variable. Let $(D_h(s))^{-1} N_h(s)$ be a left coprime factorization of $H(s)$. When $m=1$ (SISO agent case), we can write
\begin{equation}\label{eq:hs}
h(s) = n(s)/d(s) \in \R_p.
\end{equation}
The matrix $A$ describes the interconnection between agents, while $B$ (resp., $C$) indicates which agents are actuated (resp., measured), namely, there is an input $u \in \RR^{Mm}$ used to control the network and an output $y \in \RR^{Mm}$ taken from this network. In this paper, we aim to design:
\begin{itemize}[leftmargin=*,nosep]
\item a state observer providing estimates $\hat{x}_i$ of the agent states using the known network input $u$ and output $y$;
\item an observer-based state-feedback controller generating $u$ that stabilizes the network using the state estimates.
\end{itemize}
We define the following matrices:
\begin{subequations}\label{eq:calABC}
\begin{align}
\A &:= I_N \otimes A_h + A \otimes (B_h C_h) \in \RR^{Nn \times Nn},\\
\B &:= B \otimes B_h \in \RR^{Nn \times Mm},\\
\C &:= C \otimes C_h \in \RR^{Mm \times Nn}.
\end{align}
\end{subequations}
Denoting $x = (x_1,x_2,\ldots,x_N) \in \RR^{Nn}$, we obtain the multi-agent system representation~\citep{haraCDC07}
\begin{equation}\label{eq:sysx}
\dot{x} = \A x + \B u,\qquad
y = \C x.
\end{equation}
Next, we study controller$\slash$observer design for system~\eqref{eq:sysx}.

\section{Observer-Based Controller Design and Separation Principle}
\subsection{Problem Formulation}
A state-feedback controller for system~\eqref{eq:sysx} takes the form
\begin{equation}\label{eq:controller}
    u = -\K x \in \RR^{M m},
\end{equation}
where $\K \in \RR^{Mm \times Nn}$ is the controller gain to design. Next, an observer for system~\eqref{eq:sysx} has the form
\begin{equation}\label{eq:observer}
\dot{\hat{x}} = \A\hat{x} + \B u + \L(y - \hat{y}), \qquad \hat{y} = \C\hat{x}, 
\end{equation}
where $\hat{x} = (\hat{x}_1,\hat{x}_2,\ldots,\hat{x}_N)\in \RR^{Nn}$ in which each $\hat{x}_i$ is the estimate of $x_i$ for $i = 1,2,\ldots,N$, and $\L \in \RR^{Nn \times Mm}$ is the observer gain to find. The observer compares the network's output with its estimate for correction.
\begin{definition}
	System~\eqref{eq:sysx} is said to be stabilizable (resp., detectable) if there exists $\K$ (resp., $\L$) such that $\A - \B\K$ (resp., $\A - \L\C$) is Hurwitz. 
\end{definition}
The following result is standard in linear systems like~\eqref{eq:sysx}.
\begin{lemma}
	There exists an exponentially stabilizing controller~\eqref{eq:controller} (resp., observer~\eqref{eq:observer}) for system~\eqref{eq:sysx} if and only if this system is stabilizable (resp., detectable).
\end{lemma}
However, finding the gains $\K$ and $\L$ is computationally heavy in large networks. In the next part, we propose a constructive method for designing such gains.

\subsection{Designs Based on Generalized Frequency Variables}
To design the gains $\K$ and $\L$ using generalized frequency variables, we rely on the following preliminary results.
\begin{lemma}[\cite{haraRobust}]\label{lemA}
Consider an autonomous network with structure matrix $A$, where each agent is described by~\eqref{eq:sysi}. Define $p(\lambda,s) := \det(sI_n - A_h - \lambda B_h C_h)$. This system is exponentially stable, i.e., $\A$ is Hurwitz, if and only if $\sigma(A) \subset \Lambda_s := \{\lambda \in \CC: p(\lambda,s) \text{ is Hurwitz}\}$.
\end{lemma}
When $m = 1$ (SISO agent), we have $p(\lambda,s) = d(s) - \lambda n(s)$ ($d(s),n(s)$ from~\eqref{eq:hs}) and the same results hold~\citep[Theorem 1]{haraFreq}.
We propose these forms for the gains:
\begin{align}
	\K & = K\otimes C_h,\label{eq:calK}\\
	\L &= L \otimes B_h,\label{eq:calL}
\end{align}
where $K \in \RR^{M \times N}$ and $L \in \RR^{N \times M}$ are \emph{distributive} gains to find. The next lemma gives a systematic way to pick $L$.
\begin{proposition}\label{propo_obs}
System~\eqref{eq:sysx} is detectable, i.e., observer~\eqref{eq:observer} exists, if there exists $L$ such that $\sigma(A-LC) \subset \Lambda_s$.
\end{proposition}
\begin{pf}
With $\L$ in~\eqref{eq:calL}, the error $\tilde{x} := x - \hat{x}$ verifies
\begin{align}
\dot{\tilde{x}} & = (\A - \L\C)\tilde{x}\notag\\ & = (I_N \otimes A_h + A \otimes (B_h C_h) - (L \otimes B_h)(C \otimes C_h))\tilde{x} \notag\\ & = (I_N \otimes A_h + A \otimes (B_h C_h) - (LC) \otimes (B_hC_h))\tilde{x}\notag\\ & = (I_N \otimes A_h + (A-LC) \otimes (B_h C_h))\tilde{x}. \label{eq:syse}
\end{align}
This system is equivalent to an autonomous network of the same agents but with interconnection matrix $A-LC$ instead of $A$. Using Lemma~\ref{lemA}, we get the result.\hfill $\blacksquare$
\end{pf}
Now, we are interested in stabilizing system~\eqref{eq:sysx}. Ideally, if the full state is available, we can use the controller~\eqref{eq:controller} with $\K$ of the form~\eqref{eq:calK}
and show that the closed-loop system is
\begin{equation}
\dot{x} = (I_N \otimes A_h + (A-BK) \otimes (B_h C_h))x.
\end{equation}
Using logic similar to the observer, we can pick $K$ such that $\sigma(A-BK) \subset \Lambda_s$, which will then stabilize $x$. The existence of such a $K$ in turn implies the stabilizability of system~\eqref{eq:sysx}. However, the number of agents is typically very large, thus $x$ is typically not fully available, and we measure only a partial $y$ from the network using sensors. We then use the observer-based feedback controller
\begin{equation}\label{eq:controller_hat}
    u = -\K \hat{x},
\end{equation}
with $\hat{x}$ coming from~\eqref{eq:observer}, thereby generalizing the well-known \emph{separation principle} to networks as follows.

\begin{proposition}\label{propo_sepa}
The concatenation~\eqref{eq:sysx}-\eqref{eq:observer} with $u$ given by~\eqref{eq:controller_hat} is exponentially stable if there exist $K$ and $L$ such that $\sigma(A-BK) \subset \Lambda_s$ and $\sigma(A-LC) \subset \Lambda_s$.
\end{proposition}
\begin{pf}
Feeding~\eqref{eq:sysx} with $u$ from~\eqref{eq:controller_hat} with $\K$ in~\eqref{eq:calK}, we get
\begin{align}
\dot{x} & = \A x -\B\K\hat{x} = (\A -\B\K)x + \B\K \tilde{x} \notag\\& = (I_N \otimes A_h + A \otimes (B_h C_h) - (B \otimes B_h) (K\otimes C_h))x\notag \\& \qquad{}+(B \otimes B_h) (K\otimes C_h) \tilde{x}\notag \\&
= (I_N \otimes A_h + A \otimes (B_h C_h)-(BK) \otimes (B_hC_h))x \notag\\
&\qquad{}+ ((BK) \otimes (B_hC_h)) \tilde{x}\notag\\
&
= (I_N \otimes A_h + (A-BK) \otimes (B_h C_h))x\notag\\
&\qquad{} + ((BK) \otimes (B_hC_h)) \tilde{x},
\end{align}
where $\tilde{x} = x - \hat{x}$ is the estimation error. Since the same $u$ is fed to observer~\eqref{eq:observer}, the dynamics of $\tilde{x}$ are still the same as~\eqref{eq:syse}. The concatenated autonomous $(x,\tilde{x})$ dynamics are then~\eqref{eq:sepa}.
\begin{figure*}[t]
\begin{equation}\label{eq:sepa}
\begin{pmatrix}
\dot{x} \\ \dot{\tilde{x}}
\end{pmatrix}
=
\begin{pmatrix}
I_N \otimes A_h + (A-BK) \otimes (B_h C_h) & (BK) \otimes (B_hC_h) \\
0 &  I_N \otimes A_h + (A-LC) \otimes (B_h C_h)
\end{pmatrix}
\begin{pmatrix}
x \\ \tilde{x}
\end{pmatrix}.
\end{equation}
\end{figure*}
Since the system matrix is block upper triangular, exponential stability holds if and only if both $I_N \otimes A_h + (A-BK) \otimes (B_h C_h)$ and $I_N \otimes A_h + (A-LC) \otimes (B_h C_h)$ are Hurwitz, independently of each other. Applying Lemma~\ref{lemA} block-wise yields the result.\hfill $\blacksquare$
\end{pf}
Note that Propositions~\ref{propo_obs} and~\ref{propo_sepa} are \emph{sufficient} conditions for stabilizability$\slash$detectability because $\K$ and $\L$ may still exist but not in the forms~\eqref{eq:calK} and~\eqref{eq:calL}. Moreover,  if $\Lambda_s$ is non-empty and $(A,B)$ is controllable (resp., $(A,C)$ is observable), then system~\eqref{eq:sysx} is stabilizable (resp., detectable). Finding the controller and observer gains this way can considerably reduce computation in large networks.

\subsection{Example: Inverted Pendulums on Carts}
To illustrate our method in stabilizing and observing the whole network by controlling and measuring from only a small number of agents, we study a network of locally controlled inverted pendulums on carts described in~\cite[Section V.B]{haraFreq}. Each SISO agent equipped with an appropriate PD controller has the transfer function
\begin{equation}\label{eq:hs_eg}
h(s) = \frac{(0.5s+1)(1.9 s^2 - 0.002 s + 2.1)}{s(s-2)(s+1)(s+5)}.
\end{equation}
It is shown in~\citep{haraFreq} that $h(s)$ has a narrow stability region $\Lambda_s$ that does not intersect with the real axis---see Fig.~\ref{fig:pend_diverge}-Left. We consider a network of $N=4$ such locally controlled agents. The four agents are assumed to communicate with each other through a gain $k$ in a cyclic network, i.e.,
$
A =
k\left(\begin{smallmatrix}
0 & 0 & 0 & 1\\
1 & 0 & 0 & 0\\
0 & 1 & 0 & 0\\
0 & 0 & 1 & 0
\end{smallmatrix}\right)
$.
With $k=10$, we see that the eigenvalues of $A$, which are on a circle of radius $k$ centered at the origin, are outside the stability region, and so the multi-agent system is unstable---see Fig.~\ref{fig:pend_diverge}-Right.
\begin{figure}[H]
	\begin{center}	
\includegraphics[width=0.494\columnwidth,height=0.25\columnwidth]{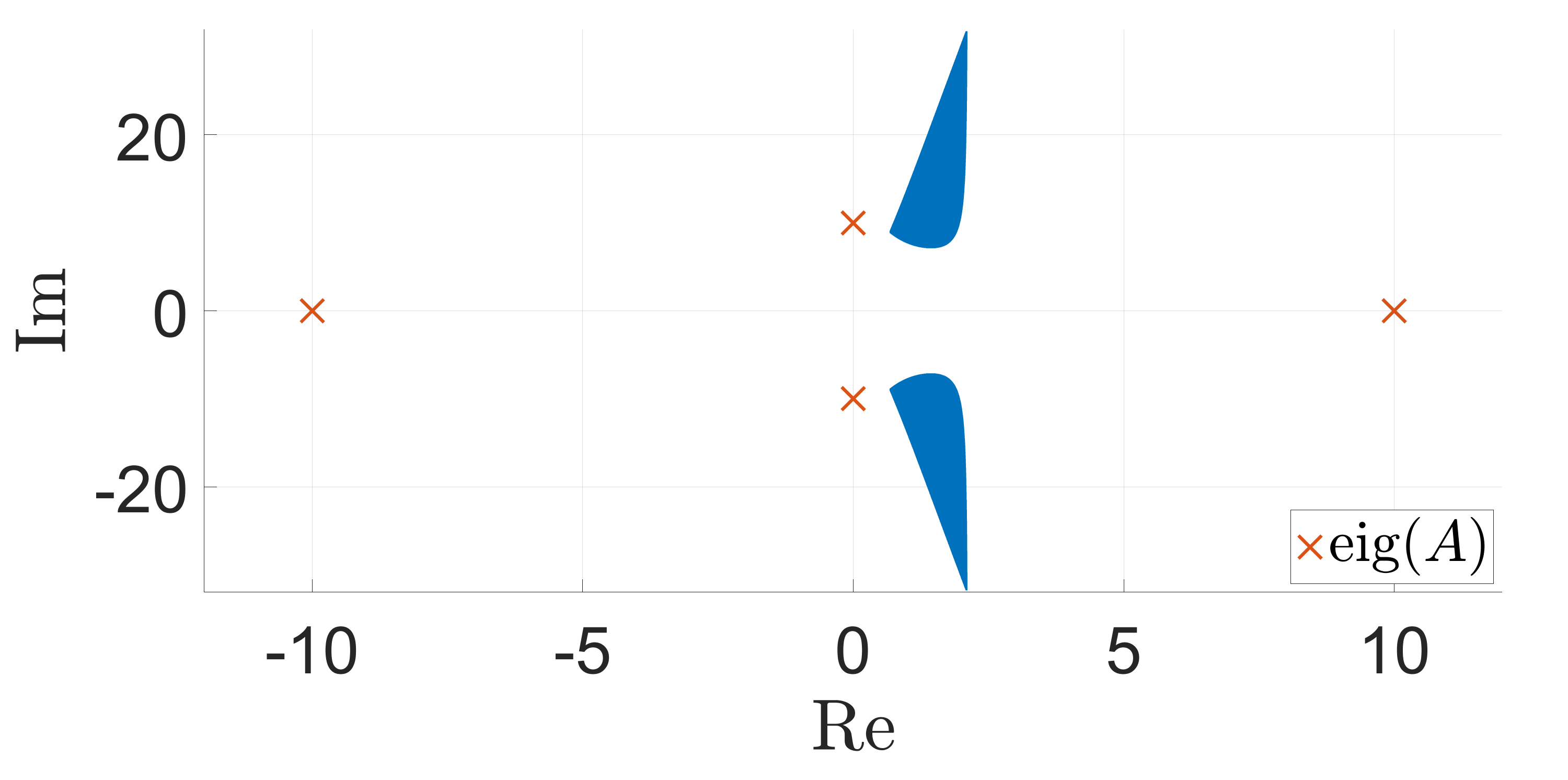} \includegraphics[width=0.494\columnwidth,height=0.25\columnwidth]{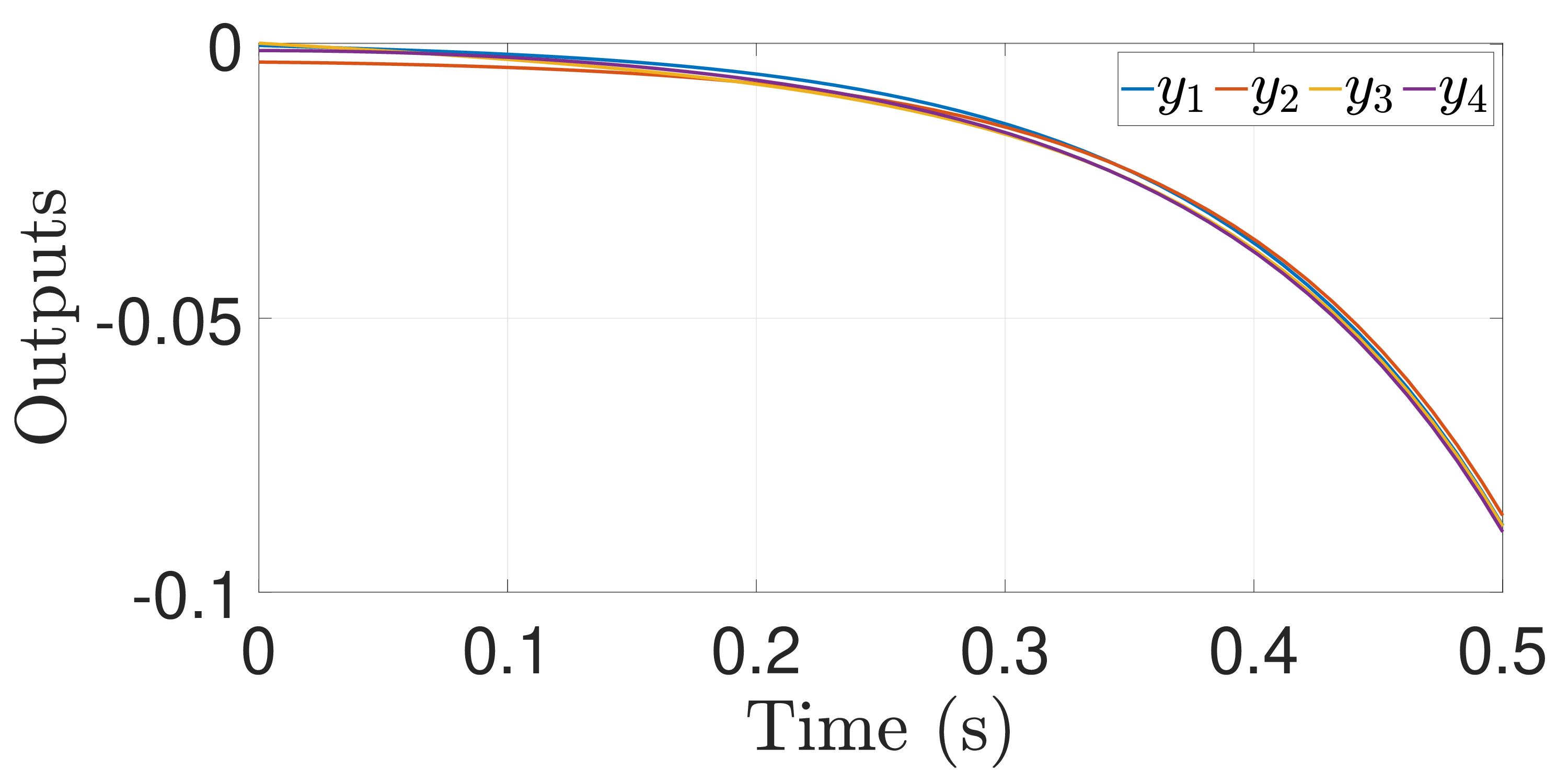} 
	\caption{Left: Stability region $\Lambda_s$ from $h(s)$ (blue) and $\sigma(A)$. Right: Unstable trajectories of uncontrolled network.}
		\label{fig:pend_diverge}
	\end{center}
\end{figure}
We consider a \emph{non-collocated} sensor–actuator configuration,
which is typically more difficult than the collocated case. The control input enters the first agent and the measurement is taken from the third one, i.e., $B = (1,0,0,0)$ and $C =
\begin{pmatrix}
0 &0&1&0
\end{pmatrix}$. Each SISO agent is minimal, and since $(A,B)$ and $(A,C)$ are respectively controllable and observable,~\cite[Proposition 3.1]{haraSICE} guarantees that the multi-agent system is controllable and observable (hence stabilizable and detectable). The controller and observer gains take the forms~\eqref{eq:calK}-\eqref{eq:calL}. Based on Proposition~\ref{propo_obs}, we pick $L$ to place the eigenvalues of $A-LC$ inside the stability region---see Fig.~\ref{fig:pend_observed}-Left. Then, using Proposition~\ref{propo_sepa}, we independently place the eigenvalues of $A-BK$ inside this region---see Fig.~\ref{fig:pend_controlled}-Left. Next, observer~\eqref{eq:observer} and controller~\eqref{eq:controller_hat} are constructed and combined.
\begin{figure}[H]
	\begin{center}	
\includegraphics[width=0.494\columnwidth,height=0.25\columnwidth]{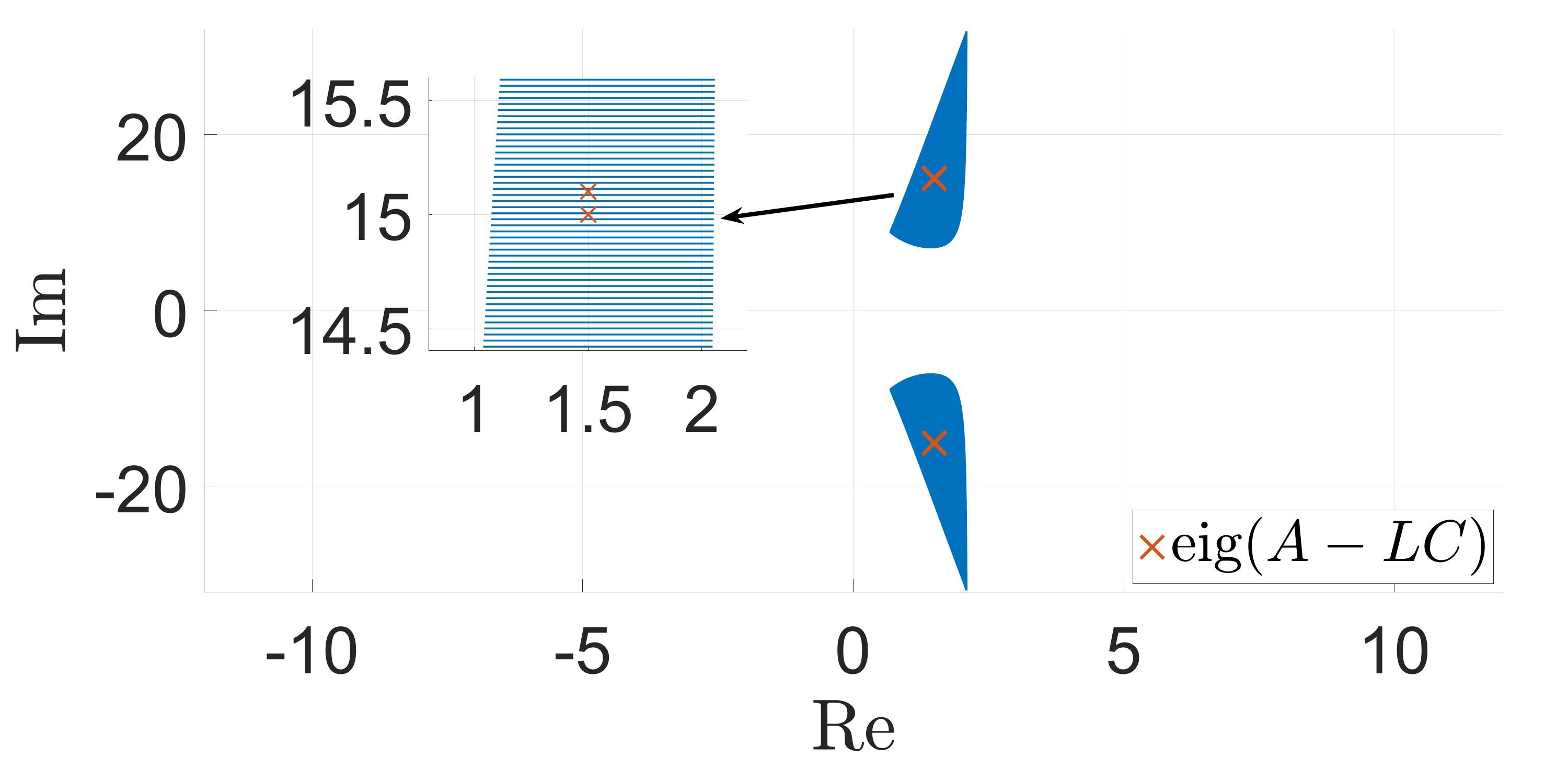} \includegraphics[width=0.494\columnwidth,height=0.25\columnwidth]{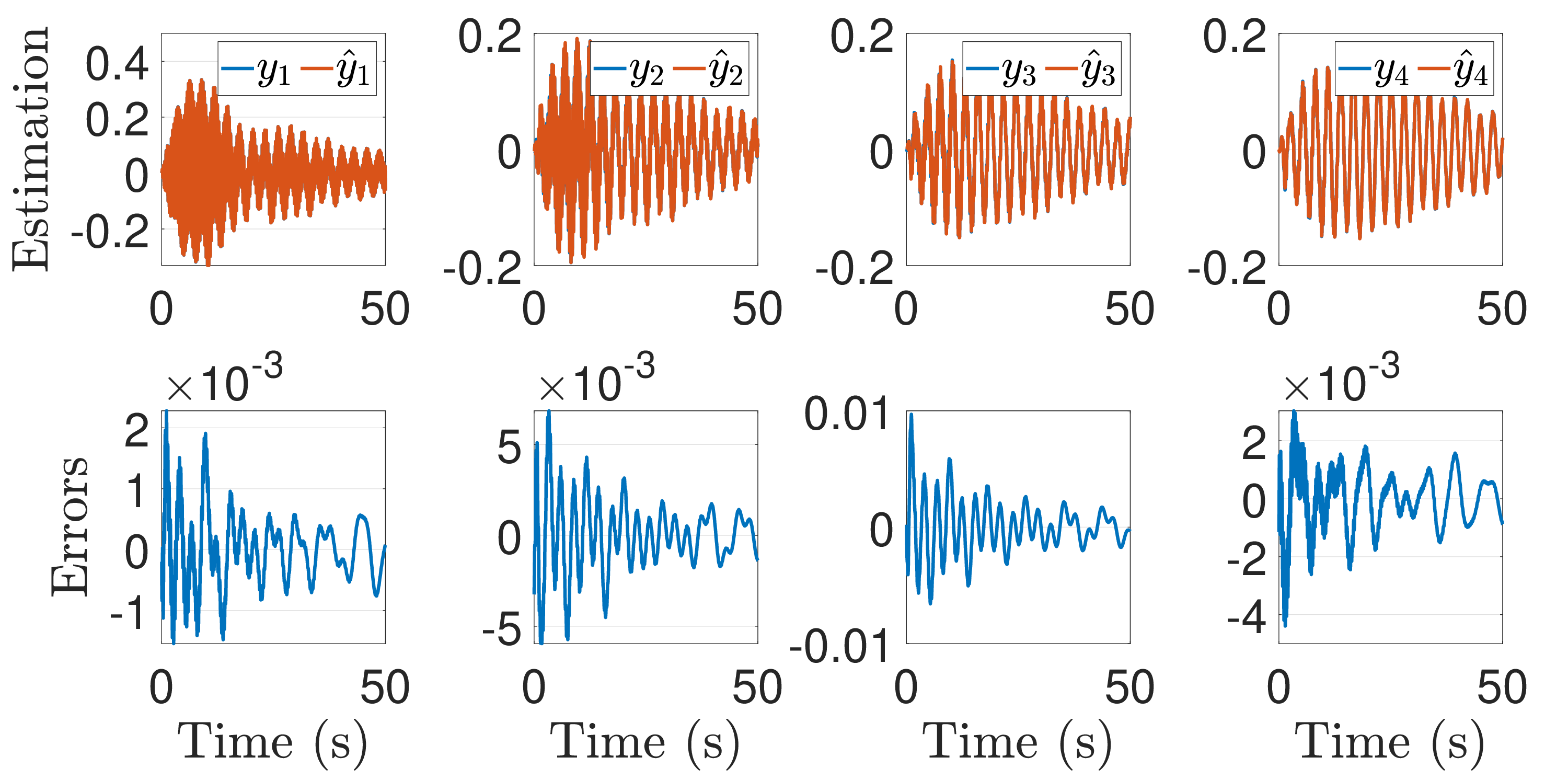} 
	\caption{Left: Stability region $\Lambda_s$ from $h(s)$ (blue) and $\sigma(A-LC)$. Right: Estimation results and errors.}
		\label{fig:pend_observed}
	\end{center}
\end{figure}
Simulation results in Fig.~\ref{fig:pend_observed}-Right and Fig.~\ref{fig:pend_controlled}-Right show that the estimation errors and trajectories are exponentially stable, though with a slow convergence rate and significant oscillations. This is because the estimation error and closed-loop dynamics include the eigenvalues with imaginary part significantly larger than the real part. Here, performance cannot be improved due to the emptiness of the \emph{$\D$-stability region} as shown in~\citep{haraFreq}.
\begin{figure}[H]
	\begin{center}	
\includegraphics[width=0.494\columnwidth,height=0.25\columnwidth]{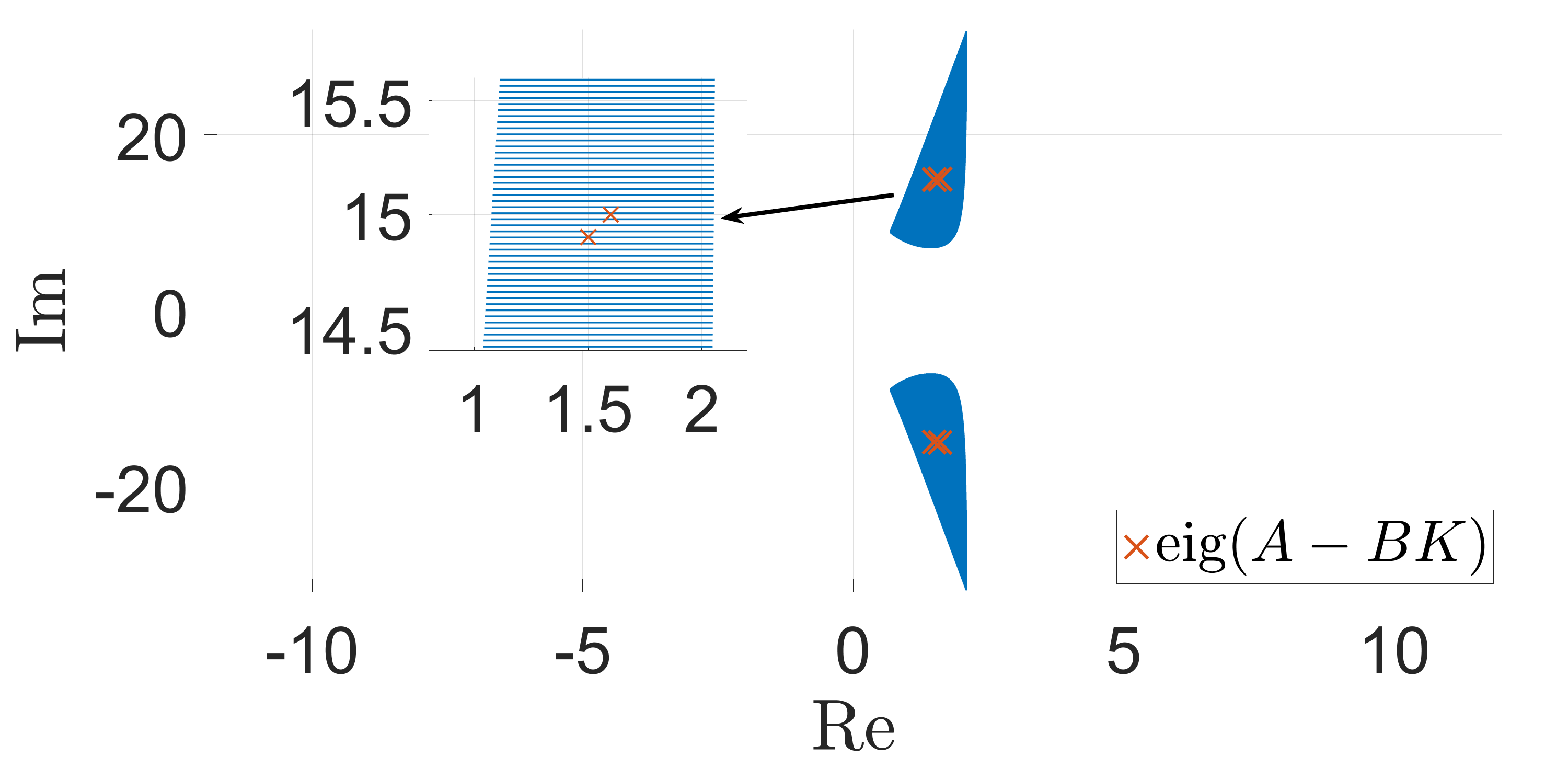} \includegraphics[width=0.494\columnwidth,height=0.25\columnwidth]{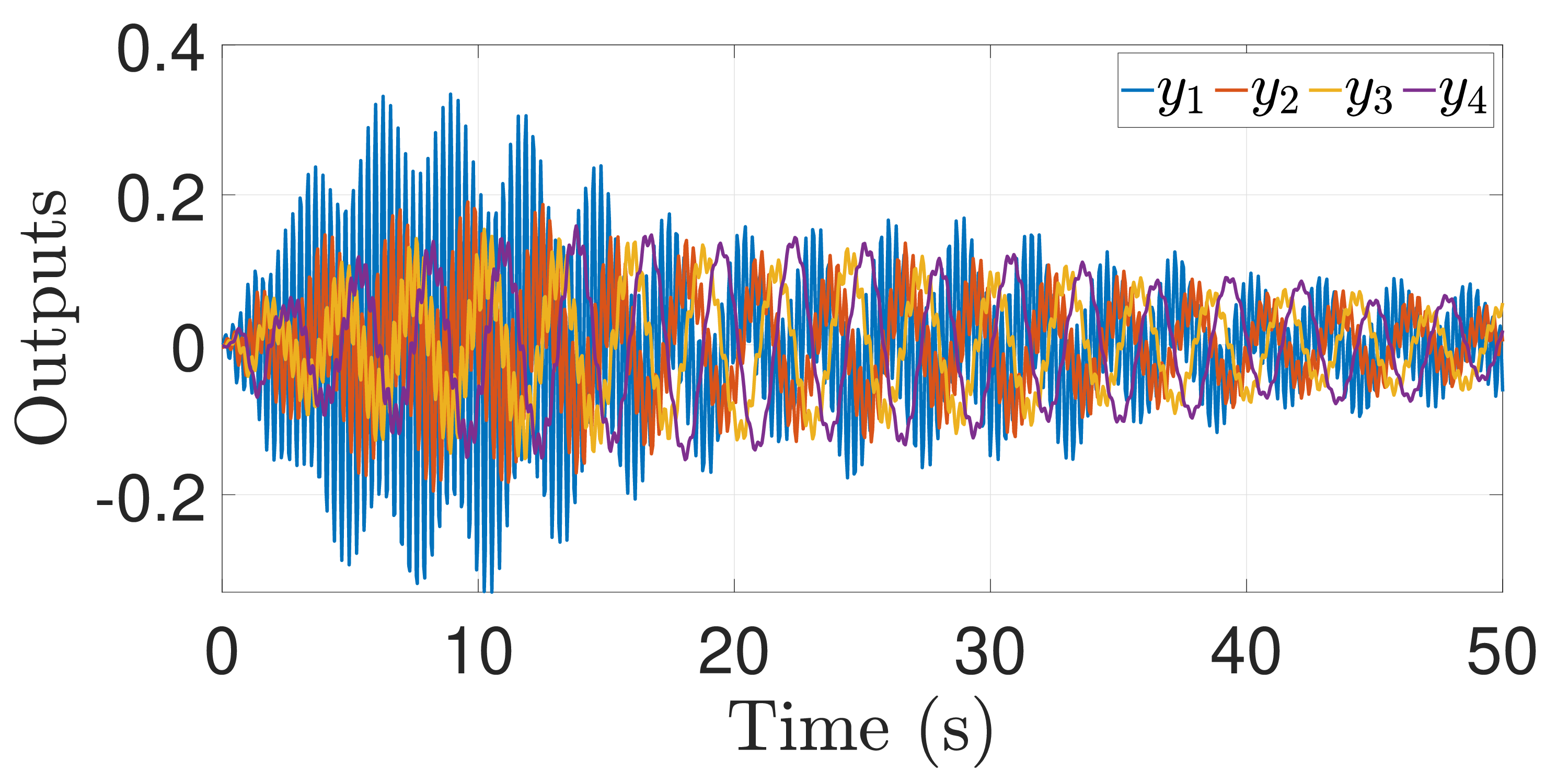} 
	\caption{Left: Stability region from $h(s)$ (blue) and $\sigma(A-BK)$. Right: Exponentially stabilized trajectories.}
		\label{fig:pend_controlled}
	\end{center}
\end{figure}

\section{On Controllability and Observability}
We discuss the controllability and observability of system~\eqref{eq:sysx}, which correspond to those of the pairs $(\A,\B)$ and $(\A,\C)$, respectively, and are sufficient conditions for controller and observer design.
Recall that in the MIMO agent case ($m > 1$),~\cite[Theorem 9]{trumpf} shows that, under the assumption
$\rank(B)<N$, the pair $(\A,\B)$ is controllable if and only if
$(A_h,B_h)$ is controllable, $(A_h,C_h)$ is observable, and the polynomial matrix
\begin{equation}\label{eq:Psi}
\Psi(s):=\begin{pmatrix}I_N \otimes D_h(s)-A\otimes N_h(s) & -B\otimes N_h(s)
\end{pmatrix}
\end{equation}
is left prime ($D_h(s),N_h(s)$ given below~\eqref{eq:Hs}). For the SISO agent case ($m = 1$),~\cite[Proposition 3.1]{haraSICE} shows that when
$\rank(B)<N$, $(\A,\B)$ is controllable if and only if
$(A_h,B_h)$ is controllable, $(A_h,C_h)$ is observable, and $(A,B)$ is controllable. Observability conditions can be stated as a duality. Now, we aim to derive alternative conditions that cover the MIMO agent case while separating the agent properties from those of their interaction structures, which are more intuitive from application viewpoints. We first show that the multi-agent system cannot be more controllable (or observable) than its structure.
\begin{lemma}\label{lem_AB}
If $(\A,\B)$ is controllable, then $(A,B)$ is controllable.
If $(\A,\C)$ is observable, then $(A,C)$ is observable.
\end{lemma}
\begin{pf}
We prove only controllability; observability follows by duality.
Assume that $(A,B)$ is uncontrollable. Then, by the PBH test (see, e.g.,~\cite{hespanha2018linear}), there exist $\lambda\in\sigma(A)$ and a non-zero $v\in\CC^N$ such that $v^\top A=\lambda v^\top$ and $v^\top B=0$.
Let $\mu\in\sigma(A_h+\lambda B_hC_h)$ and pick a non-zero left eigenvector $\eta\in\CC^n$ such that $\eta^\top(A_h+\lambda B_hC_h)=\mu \eta^\top$. Define $z^\top:=v^\top\otimes \eta^\top\in\CC^{1\times Nn}$, which is non-zero. Using Kronecker identities, we get
\begin{align*}
z^\top \A
&=(v^\top\otimes \eta^\top)(I_N\otimes A_h + A\otimes(B_hC_h))\\
&=v^\top\otimes(\eta^\top A_h) + (v^\top A)\otimes(\eta^\top B_hC_h)\\
&=v^\top\otimes(\eta^\top A_h) + (\lambda v^\top)\otimes(\eta^\top B_hC_h)\\
&=v^\top\otimes (\eta^\top(A_h+\lambda B_hC_h))
=\mu (v^\top\otimes \eta^\top)=\mu z^\top,\\
z^\top \B&=(v^\top\otimes \eta^\top)(B\otimes B_h)\\&=(v^\top B)\otimes(\eta^\top B_h)=0\otimes(\eta^\top B_h)=0.
\end{align*}
Thus, there exists a left eigenvector $z^\top$ of $\A$ such that $z^\top\B=0$. So, $(\A,\B)$ is uncontrollable by the PBH test. \hfill $\blacksquare$
\end{pf}
We then find necessary conditions for network controllability and observability from~\citep{trumpf}.
\begin{lemma}\label{lem_con_obs}
Assume that $\rank(B_h)=\rank(C_h)=m$.\footnote{In the SISO agent case, this assumption is already included in the controllability of $(A_h,B_h)$ and the observability of $(A_h,C_h)$.}
\begin{enumerate}[label=(\roman*),leftmargin=20pt]
\item (Controllability) Assume that $\rank(B) < N$.
Then, $(\A,\B)$ is controllable only if
\begin{enumerate}[label=(\alph*),leftmargin=*]
\item $(A_h,B_h)$ is controllable; $(A_h,C_h)$ is observable;
\item For every $\lambda\in\sigma(A)$ such that there exists a non-zero vector
$v\in\CC^N$ satisfying 
$v^\top \begin{pmatrix}A-\lambda I_N & B\end{pmatrix} = 0$, 
we have $\det(D_h(s)-\lambda N_h(s))\neq 0$ for all $s\in\CC$.
\end{enumerate}
\item (Observability) Assume that $\rank(C) < N$.
Then, $(\A,\C)$ is observable only if
\begin{enumerate}[label=(\alph*)]
\item $(A_h,B_h)$ is controllable; $(A_h,C_h)$ is observable;
\item For every $\lambda\in\sigma(A)$ such that there exists a non-zero vector
$w\in\CC^N$ satisfying $\begin{pmatrix}
A - \lambda I_N\\C
\end{pmatrix}w = 0$,
we have $\det(D_h(s)-\lambda N_h(s))\neq 0$ for all $s\in\CC$.
\end{enumerate}
\end{enumerate}
\end{lemma}
\begin{pf}
We prove only (i); (ii) follows by duality. When $\rank(B)\neq N$,~\cite[Theorem 9]{trumpf} shows that $(\A,\B)$ is controllable if and only if (i)(a) holds and the polynomial matrix $\Psi(s)$ in~\eqref{eq:Psi} is left prime for all complex $s$, which holds if and only if it has full row rank for all $s \in \CC$, so (i)(b) is equivalent to $\rank(\Psi(s))=Nm$ for all $s\in\CC$.
Thus, under (i)(a), the pair $(\A,\B)$ is controllable if and only if $\rank(\Psi(s))=Nm$ for all $s\in\CC$. Assume that there exist $\lambda\in\sigma(A)$ and a non-zero $v\in\CC^N$ such that
$
v^\top A=\lambda v^\top$, $v^\top B=0
$, and there exists $s_0\in\CC$ such that $\det(D_h(s_0)-\lambda N_h(s_0))=0$.
Then, there exists a non-zero $\eta\in\CC^m$ such that $
\eta^\top(D_h(s_0)-\lambda N_h(s_0))=0$. Define the non-zero row vector $z^\top:=(v^\top\otimes \eta^\top)\in\CC^{1\times Nm}$. Using Kronecker identities, we get
\begin{align*}
&z^\top(I_N\otimes D_h(s_0)-A\otimes N_h(s_0))\\
&=(v^\top\otimes\eta^\top)(I_N\otimes D_h(s_0)-A\otimes N_h(s_0))\\
&=v^\top\otimes(\eta^\top D_h(s_0))-(v^\top A)\otimes(\eta^\top N_h(s_0))\\
&=v^\top\otimes(\eta^\top D_h(s_0))-\lambda v^\top \otimes(\eta^\top N_h(s_0))\\
&=v^\top\otimes(\eta^\top(D_h(s_0)-\lambda N_h(s_0)))=v^\top\otimes0=0,
\end{align*}
and similarly
\begin{align*}
&z^\top(B\otimes N_h(s_0))
=(v^\top\otimes\eta^\top)(B\otimes N_h(s_0))\\&=(v^\top B)\otimes(\eta^\top N_h(s_0))=0\otimes(\eta^\top N_h(s_0))=0.
\end{align*}
Therefore, $z^\top\Psi(s_0)=0$, which implies $\rank(\Psi(s_0))<Nm$, hence $\Psi(s_0)$ is not left prime, and so $(\A,\B)$ is uncontrollable. This proves that (i)(b) is necessary.\hfill $\blacksquare$
\end{pf}
If we combine Lemmas~\ref{lem_AB} and~\ref{lem_con_obs}, controllability of $(\A,\B)$ implies that of $(A,B)$, and so (i)(b) follows. While the criterion of~\citep{trumpf} provides a complete necessary and sufficient condition for controllability of MIMO networks through the left-primeness of the polynomial matrix $\Psi(s)$, this condition is expressed in frequency-domain and polynomial-matrix terms that are difficult to interpret in classical state-space language. In contrast, for SISO agents,~\citep{haraSICE} shows that this condition reduces to the simple PBH test for the interconnection pair $(A,B)$ under agent minimality. The situation is fundamentally different in the MIMO case. In general, controllability of $(A,B)$ together with minimality of the agent is \emph{not} sufficient for controllability of the network, and the mechanism by which controllability is lost is not directly visible from the PBH test, as seen next. 

\begin{example}
Consider an example of two agents described by
$
A_h=\left(\begin{smallmatrix}
1&-2&1\\
0&1&-1\\
2&-2&0
\end{smallmatrix}\right)$, 
$B_h=\left(\begin{smallmatrix}
-1&2\\
2&2\\
0&2
\end{smallmatrix}\right)$,
$C_h=\left(\begin{smallmatrix}
-2&0&2\\
0&0&2
\end{smallmatrix}\right)$,
connected via the structure 
$
A=\left(\begin{smallmatrix}
2&2\\
-1&-1
\end{smallmatrix}\right)$, $B=\left(\begin{smallmatrix}
-1\\
2
\end{smallmatrix}\right)$. Here $n = 3$, $m = 2$, $N = 2$, $M = 1$.
One can verify that: $(A_h,B_h)$ is controllable and $(A_h,C_h)$ is observable with $\rank(B_h) = \rank(C_h) = 2 = m$, and $(A,B)$ is controllable with $\rank(B)= 1 < 2 = N$. However, $(\A,\B)$ given by~\eqref{eq:calABC} is uncontrollable---its controllability matrix has rank $5$ (instead of $Nn = 6$). This shows that unlike the SISO agent case, in this MIMO context, agent minimality and structure controllability are not sufficient for network controllability. In this example, uncontrollability of $(\A,\B)$ does not arise from a PBH obstruction of $(A,B)$, since $(A,B)$ is controllable. To illustrate this, we look for $z = (z_1,z_2) \in \RR^6$ where each $z_i \in \RR^3$ such that $z^\top\B
= (-z_1 + 2z_2)^\top B_h
=\begin{pmatrix}0&0\end{pmatrix}$. One choice is $z_1=(-8,1,7)$ and $z_2=(-7,-1,8)$. Hence, $z^\top\B=0$ although $z_1$ and $z_2$ are not proportional. Moreover, $z^\top \A = 3z^\top$, so $z$ certifies the uncontrollability of $(\A,\B)$. This cancellation occurs in $\RR^{1\times 2}$ and is genuinely
directional. It cannot be detected by PBH for $(A,B)$ and is
impossible in the SISO case ($m=1$) where $z_i^\top B_h$ is a scalar, so a purely MIMO mechanism.
\end{example}

\begin{ack}
The work of G. Q. B. Tran was supported in part by the AFOSR MURI FA9550-23-1-0337 grant. His visit to Keio University in Jan.-Mar. 2025 for this project was funded by the Honda Y-E-S Award Plus from Honda Foundation.
\end{ack}

\section*{DECLARATION OF GENERATIVE AI AND AI-ASSISTED TECHS. IN THE WRITING PROCESS}
During the preparation of this work, the authors used ChatGPT to improve language. After using this tool, the authors reviewed and edited the content as needed and take full responsibility for the content of the publication.

\bibliography{ref} 
\end{document}